\documentclass[oneside]{ipart_v1}

\usepackage{physics2,derivative,amsmath,amssymb,cancel,color}
\usephysicsmodule{ab,ab.braket}
\usepackage{geometry}
\usepackage{hyperref}
\DeclareMathOperator*{\supp}{supp}
\DeclareMathOperator*{\esssup}{ess\,sup}
\DeclareMathOperator*{\essinf}{ess\,inf}
\DeclareMathOperator*{\esssupp}{ess\,supp}

\theoremstyle{plain}
\newtheorem{theorem}{Theorem}[section]
\newtheorem{lemma}[theorem]{Lemma}
\newtheorem{corollary}[theorem]{Corollary}
\theoremstyle{definition}
\newtheorem{definition}[theorem]{Definition}

\begin{document}


\title[Symmetry of Bounce Solutions at Finite Temperature]{Symmetry of Bounce Solutions\\at Finite Temperature}
\author[Yutaro Shoji and Masahide Yamaguchi]{Yutaro Shoji\footnotemark[1]\footnotemark[2] and Masahide Yamaguchi\footnotemark[3]\footnotemark[4]\footnotemark[5]}
\dedicatory{\footnotemark[1]Jo\v{z}ef Stefan Institute, Jamova 39, 1000 Ljubljana, Slovenia\\
\footnotemark[2]Centre For Cosmology and Science Popularization (CCSP), SGT University, Gurugram, Delhi-NCR, Haryana 122505, India\\
\footnotemark[3]Cosmology, Gravity, and Astroparticle Physics Group, Center for Theoretical Physics of the Universe, Institute for Basic Science (IBS), Daejeon, 34126, Korea\\
\footnotemark[4]Department of Physics, Institute of Science Tokyo, 2-12-1 Ookayama, Meguro-ku, Tokyo 152-8551, Japan\\
\footnotemark[5]Department of Physics and IPAP, Yonsei University, 50 Yonsei-ro, Seodaemun-gu, Seoul 03722, Korea\\
}

\begin{abstract}
    The seminal work of Coleman, Glaser, and Martin established that, at zero temperature, any non-trivial solution to the equations of motion with the least Euclidean action is $O(D)$-symmetric. This paper extends their foundational analysis to finite temperature. We rigorously prove that for a broad class of scalar potentials, any saddle-point configuration with the least action is necessarily $O(D\!-\!1)$-symmetric and monotonic in the spatial directions. This result provides a firm mathematical justification for the symmetry properties widely assumed in studies of thermal vacuum decay and cosmological phase transitions.
\end{abstract}

\maketitle

\section{Introduction}

\label{sec:Intro}

Vacuum decay describes the transition of a system from a metastable (false) vacuum to a more stable (true) vacuum state through quantum or thermal fluctuations. This phenomenon is crucial in diverse fields, including early universe cosmology, where it underpins scenarios of phase transitions and the generation of gravitational waves, as well as in the study of string landscapes. The semi-classical approach to these tunneling effects hinges on identifying bounce solutions: non-trivial saddle points of the Euclidean action that connect the false and true vacua \cite{Coleman:1977py,Callan:1977pt}. The probability of such a decay, expressed as a decay rate per unit volume, is exponentially suppressed by the Euclidean action of these bounce solutions.

At zero temperature, the pioneering work of Coleman, Glaser, and Martin (CGM) \cite{Coleman:1977th} rigorously demonstrated that any saddle-point configuration with the least Euclidean action exhibits full $O(D)$ spherical symmetry in $D$-dimensional Euclidean spacetime, assuming the absence of gravitational effects\footnote{The $O(D)$ symmetry of the bounce solution has also been established in more general settings, such as multi-field models \cite{cmp/1103941720,LOPES1996378,Byeon2008SymmetryAM,Blum:2016ipp}. More recently, a proof including gravitational effects was provided under specific conditions using the AdS/CFT correspondence \cite{Oshita:2023pwr}.}.
This inherent symmetry significantly simplifies the analysis, as the bounce solution then depends solely on the radial coordinate $r = \sqrt{\tau^2 + |x|^2}$, thereby reducing the equations of motion to a set of ordinary differential equations.

However, the situation changes drastically at finite temperature. As first shown by Linde \cite{Linde:1980tt,Linde:1981zj} (see also \cite{Affleck:1980ac,Garriga:1994ut}), thermal field theory requires compactifying the Euclidean time direction with periodicity $\beta = 1/T$, transforming the spacetime manifold into $\mathbb{R}^{D-1} \times S^1$. The bounce solution must now be periodic in the Euclidean time direction and the $O(D)$-symmetry should be broken: Instead of $O(D)$ symmetry, the system at finite temperature admits at most $O(D-1)$ symmetry.
While the CGM argument can be adapted to the high-temperature limit, where the bounce solution becomes time-independent and the problem reduces to an analysis on $\mathbb{R}^{D-1}$, it does not extend to arbitrary finite temperatures.
Consequently, the assumption of $O(D-1)$ symmetry, while lacking a general proof, has been a foundational prerequisite for various numerical techniques developed to find bounce solutions at finite temperature \cite{Ferrera:1995gs,Widyan:1998wa,Widyan:1999zg,Fan:2024txm}.

This paper establishes a rigorous mathematical foundation for the symmetry of bounce solutions in thermal field theory at arbitrary finite temperatures. We prove that for any admissible scalar potential in $D>3$ dimensions, at least one non-trivial saddle point of the Euclidean action exists, and all saddle points with the least action are necessarily spherically symmetric and monotonic in the spatial directions.
Our methodology extends the seminal work of CGM for the zero-temperature case \cite{Coleman:1977th}: We first reformulate the problem of finding non-trivial saddle points into an equivalent, genuine minimization problem. By analyzing the minimizing sequence for this reduced problem, we then prove both the existence of a minimum and the symmetry of the resulting solution.

The transition to a finite-temperature setting, which involves the compactification of the Euclidean time direction, introduces technical challenges not present in the original CGM analysis. For instance, the field configuration can no longer be reduced to a simple one-dimensional radial profile, and consequently, several arguments in the original CGM analysis do not apply. In addition, the partial symmetrization in the spatial directions, known as Steiner symmetrization, demands a more sophisticated treatment, as the kinetic term, $(\partial_\tau\phi)^2$, and the spatial gradient term, $|\nabla_x\phi|^2$, behave differently under this symmetrization.

To address these challenges, we leverage advancements in the calculus of variations. The P\'olya-Szeg\"o inequality for Steiner symmetrization, a long-established result (see \cite{256c4953-55e2-3c00-93f8-ff3cdc00211f} and references therein), provides the necessary analytical tool. Crucially, recent characterizations of the equality conditions for this inequality, as studied in \cite{CIANCHI2006673} for codimension one and in \cite{Capriani2011TheSR} for arbitrary codimensions, are instrumental in our analysis. By extending their theorems to accommodate functions that vanish at spatial infinity, we apply these results to our reduced problem to rigorously establish the symmetry of the solution.

This work has significant implications not only for the theoretical understanding of vacuum decay at finite temperature but also for practical applications in cosmology, particularly in modeling first-order phase transitions and thermal bubble nucleation in the early universe. The established spherical symmetry and monotonicity of bounce solutions directly influence the decay rate formula and the most probable shape of nucleating bubbles, which in turn affect observable signatures such as stochastic gravitational wave backgrounds. Our results therefore bridge the gap between the physical intuition regarding thermal tunneling and the rigorous mathematical underpinnings of the associated variational problems.

The remainder of this paper is structured as follows. Section~\ref{sec:Statement} formally states the central problem: identifying non-trivial saddle points of the Euclidean action in a scalar field theory at finite temperature. In Section~\ref{sec:Reduced}, we introduce a scale-invariant functional and define the equivalent reduced minimization problem. Section~\ref{sec:symmetry} is dedicated to proving the existence and symmetry of solutions to this reduced problem, rigorously demonstrating that all solutions are spherically symmetric and monotonic in the spatial directions. Finally, Section~\ref{sec:Conclusion} summarizes our findings and outlines avenues for future research.

\section{Statement of main result}
\label{sec:Statement}
This section formally presents the central problem and outlines our main theorem.
For clarity, we note that the symbol $\subset$ in this paper will denote inclusion, which is not necessarily strict.

Throughout this paper, we adopt the notation $z=(x,\tau)\in\mathbb R^{D-1}\times S^1$, where $S^1$ denotes a circle of circumference $\beta$ and is identified with the interval $[0,\beta)$. Unless explicitly specified, the integral domain is $\mathbb R^{D-1}\times S^1$, and the integration measure is defined as
\begin{equation}
    \int\odif[order=D]{z}=\int_0^\beta\odif{\tau}\int\odif[order=D-1]{x}.
\end{equation}

Our analysis focuses on spacetime dimensions $D>3$ and functions in the locally Sobolev space $W_{\rm loc}^{1,2}(\mathbb R^{D-1}\times S^1)$ that vanish at spatial infinity in the sense that
\begin{equation}
    \forall \epsilon>0,~\exists~{\rm compact}~K\subset \mathbb R^{D-1}~{\rm s.t.}~|\phi(z)|<\epsilon~\text{\rm for }\mathcal H^D\text{\rm -a.e. } z\not\in (K\times S^1).\label{eq:vanish_at_infinity}
\end{equation}
Here, $\mathcal H^k$ is the $k$-dimensional Hausdorff measure.
The action is defined on this function space as
\begin{equation}
    S[\phi]=\int\odif[order=D]{z}\left(\frac{1}{2}(\partial_\tau\phi)^2+\frac12|\nabla_x\phi|^2+V(\phi)\right),
\end{equation}
where the potential $V:\mathbb R\to\mathbb R$ is a Borel measurable function, $\nabla_x$ and $\partial_\tau$ are understood as weak derivatives with respect to $x$ and $\tau$, respectively.
We assume, without loss of generality, that the false vacuum is located at $\phi_{\rm FV}=0$ and $V(\phi_{\rm FV})=0$.
We do not assume the existence of a global (true) vacuum; nevertheless,
the potential is required to satisfy the growth conditions stated below.
\begin{definition}
    A function $V:\mathbb R\to\mathbb R$ is said to be {\it admissible} if
\begin{itemize}
    \item $V(0)=0$;
    \item $V$ is Borel measurable;
    \item $V$ is lower semicontinuous;
    \item There exists $\phi_*\in \mathbb R$ such that $V(\phi_*)<0$ and $\sup_{\phi/\phi_* \in(0,1)} V(\phi)<\infty$;
    \item There exist positive constants, $\frac{2D}{D-2}>B>A>0$, $a>0,b>0$ such that
    \begin{equation}
        V(\phi)-a|\phi|^A+b|\phi|^B\geq0,\label{eq:condition_v}
    \end{equation}
    for all $\phi\in \mathbb R$. Here, these constants may differ for $\phi<0$ and $\phi>0$.
\end{itemize}
\end{definition}

\begin{theorem}[Main theorem]
    For an admissible potential and for $D>3$, there exists at least one non-trivial saddle point of the action, 
and all saddle points with the least action are $\mathcal H^D$-a.e.~spherically symmetric and monotonic in the spatial directions, provided that their spatial derivatives are $\mathcal H^D$-a.e.~non-zero.
\end{theorem}
Note that, for an analytic $V$, the condition on the spatial derivatives is satisfied automatically as bounce solutions are also analytic.
The condition on the spatial derivatives can be relaxed to permit plateaus where the field value is extremal, {\it i.e.}, at $\phi=0$ or at the maximum of $\phi$ in the spatial directions.
This theorem is proven by Theorem \ref{thm:reduced_problem} and Theorem \ref{thm:symmetry} in the following sections.

For a differentiable potential $V$, a saddle point is, by definition, a weak solution to the Euler-Lagrange equation:
\begin{equation}
    \int\odif[order=D]{z}\ab[(\partial_\tau \delta\phi)(\partial_\tau\phi)+(\nabla_x \delta\phi)(\nabla_x\phi)+\delta\phi \,V'(\phi)]=0,
\end{equation}
for all $\delta\phi\in C_c^\infty(\mathbb R^{D-1}\times S^1)$. We focus on weak solutions throughout this paper, as they are the fundamental objects in the path integral formulation and can accommodate non-smooth configurations. The question of regularity and the existence of strong solutions is beyond the scope of this work.

\section{Reduced problem}
\label{sec:Reduced}

In this section, we introduce a scale-invariant functional and a genuine minimization problem that is equivalent to finding a non-trivial saddle point of the action. This reformulation is also necessary to utilize the symmetrization techniques that preserve the integral of the potential.

We first decompose the Euclidean action into gradient, kinetic and potential terms as
\begin{equation}
    S[\phi]=\mathcal T[\phi]+\mathcal K[\phi]+\mathcal V[\phi],
\end{equation}
where
\begin{align}
    \mathcal T[\phi]&=\int\odif[order=D]{z}\frac12|\nabla_x\phi|^2,\\
    \mathcal K[\phi]&=\int\odif[order=D]{z}\frac12(\partial_\tau\phi)^2,\\
    \mathcal V[\phi]&=\int\odif[order=D]{z}V(\phi).
\end{align}

Consider the scale transformation $\phi_\lambda(x,\tau) = \phi(x/\lambda, \tau)$ for a positive constant $\lambda$. The action of the rescaled function is given by
\begin{equation}
    S[\phi_\lambda]=\lambda^{D-3}\mathcal T[\phi]+\lambda^{D-1}\ab(\mathcal V[\phi]+\mathcal K[\phi]).
\end{equation}
If $\bar\phi$ is a saddle point of the action, it must be stationary under this scale transformation. Taking the derivative with respect to $\lambda$ and setting $\lambda=1$ yields a variant of Derrick's relation \cite{Derrick:1964ww}:
\begin{equation}
    (D-3)\mathcal T[\bar\phi]+(D-1)\ab(\mathcal V[\bar\phi]+\mathcal K[\bar\phi])=0.\label{eq:derrick}
\end{equation}
Since $\mathcal T[\bar\phi]$ is strictly positive for any non-trivial solution, this relation implies that $\mathcal V[\bar\phi]+\mathcal K[\bar\phi]<0$ and, consequently, $\mathcal V[\bar\phi]<0$. Note that if $\mathcal T[\bar\phi]=0$, then $\nabla_x\bar\phi=0$ $\mathcal H^D$-a.e., which, combined with the condition of vanishing at spatial infinity, forces the solution to be the trivial one, $\bar\phi=0$ $\mathcal H^D$-a.e.

This observation motivates the definition of the following scale-invariant functional:
\begin{equation}
    \mathcal R[\phi]=\frac{(\mathcal T[\phi])^{\frac{D-1}{D-3}}}{-\mathcal V[\phi]-\mathcal K[\phi]}.\label{eq:def_R}
\end{equation}
Note that for the admissible potential, there exists $\phi\in W^{1,2}_{\rm loc}(\mathbb R^{D-1}\times S^1)$ vanishing at spatial infinity such that $\mathcal T[\phi]<\infty$ and $\mathcal V[\phi]+\mathcal K[\phi]<0$.

\begin{definition}
The reduced problem is to find a minimum of $\mathcal R[\phi]$ for a fixed negative value of $\mathcal V[\phi]$ and for $\mathcal K[\phi]<-\mathcal V[\phi]$ among functions $\phi\in W^{1,2}_{\rm loc}(\mathbb R^{D-1}\times S^1)$ vanishing at spatial infinity.
\end{definition}

\begin{definition}
    An infinite sequence $\{\phi_n\}$ is said to be a {\it minimizing sequence} of $\mathcal R$ if
    \begin{equation}
        \lim_{n\to\infty}\mathcal R[\phi_n]=\inf_{\phi}\mathcal R[\phi],
    \end{equation}
    where infimum is taken over all $\phi\in W^{1,2}_{\rm loc}(\mathbb R^{D-1}\times S^1)$ that vanish at spatial infinity and satisfy $\mathcal V[\phi]<0$ and $\mathcal K[\phi]<-\mathcal V[\phi]$.
\end{definition}
\begin{theorem}\label{thm:reduced_problem}
    If a solution of the reduced problem exists, it is a saddle point of the action that has action less than or equal to that of any non-trivial saddle point.
\end{theorem}

\begin{proof}
Suppose that $\mathcal R[\phi_1]\leq\mathcal R[\phi_2]<\infty$ for two functions $\phi_1$ and $\phi_2$ having negative $\mathcal V+\mathcal K$. They can be rescaled to satisfy Eq.~\eqref{eq:derrick} with the scale factors given by
\begin{equation}
    \lambda_i^2=\frac{D-3}{D-1}\frac{\mathcal T[\phi_i]}{-\mathcal V[\phi_i]-\mathcal K[\phi_i]},
\end{equation}
with $i=1,2$.
The rescaled functions satisfy the following inequality:
\begin{align}
    S[\phi_{1,\lambda_1}]&=\lambda_1^{D-3}\ab[\mathcal T[\phi_1]+\lambda_1^2(\mathcal V[\phi_1]+\mathcal K[\phi_1])]\nonumber\\
    &=\ab(\frac{D-3}{D-1})^{\frac{D-3}{2}}\frac{2}{D-1}\mathcal R^{\frac{D-3}{2}}[\phi_1]\nonumber\\
    &\leq\ab(\frac{D-3}{D-1})^{\frac{D-3}{2}}\frac{2}{D-1}\mathcal R^{\frac{D-3}{2}}[\phi_2]\nonumber\\
    &=\lambda_2^{D-3}\ab[\mathcal T[\phi_2]+\lambda_2^2(\mathcal V[\phi_2]+\mathcal K[\phi_2])]\nonumber\\
    &=S[\phi_{2,\lambda_2}].
\end{align}
This demonstrates that minimizing the functional $\mathcal R$ is equivalent to minimizing the action $S$ subject to Derrick's relation, Eq.~\eqref{eq:derrick}.
Since this relation is one of the extremization conditions for the action and can always be solved by adjusting the scale parameter $\lambda$, a minimum of $\mathcal R$ corresponds to an extremum of $S$. This extremum is specifically a saddle point with an unstable direction because
\begin{equation}
    \pdv[order=2]{}{\lambda}S[\bar\phi_\lambda]|_{\lambda=1}=-2(D-3)\mathcal T[\bar\phi]<0.
\end{equation}

Now, let $\phi_R$ be a solution to the reduced problem rescaled to satisfy Derrick's relation, and let $\phi_S$ be any other non-trivial saddle point of the action. As a non-trivial saddle point, $\phi_S$ must satisfy $\mathcal{V}[\phi_S] < 0$. We can rescale $\phi_R$ to a new function $\phi'_R$ such that $\mathcal{V}[\phi'_R] = \mathcal{V}[\phi_S]$, a transformation that leaves $\mathcal{R}[\phi_R]$ invariant. By definition of $\phi_R$ as a minimizer, it follows that $\mathcal{R}[\phi'_R] = \mathcal{R}[\phi_R] \leq \mathcal{R}[\phi_S]$. The preceding analysis then implies that $S[\phi_R] \leq S[\phi_S]$. Therefore, any solution to the reduced problem corresponds to a saddle point of the action with action less than or equal to that of any other non-trivial saddle point.
\end{proof}

\section{Existence and symmetry}

\label{sec:symmetry}

This section is dedicated to establishing the existence and symmetry properties of solutions to the reduced problem. Our analysis extends the original CGM framework but requires addressing several technical challenges unique to the finite-temperature setting. Specifically, we must contend with three key differences: the scale-invariant functional is distinct from the CGM's; Steiner symmetrization may not combine disconnected support regions; and the functions under study are not necessarily absolutely continuous. These distinctions necessitate a more intricate analysis and supplementary proofs beyond those in the original CGM work.
\begin{theorem}\label{thm:symmetry}
    There exists at least one solution to the reduced problem, and all solutions are necessarily $\mathcal H^D$-a.e.~spherically symmetric and monotonic in the spatial directions, provided that their spatial derivatives are $\mathcal H^D$-a.e.~non-zero.
\end{theorem}
This theorem is established through a series of lemmas presented in the following subsections.

\subsection{Definite sign}
We first demonstrate that the minimization problem can be restricted to functions of a definite sign. Our argument adapts the logic of Coleman, Glaser, and Martin \cite{Coleman:1977th} to our scale-invariant functional $\mathcal{R}$. Subsequently, we will assume, without loss of generality, that the functions in any minimizing sequence are non-negative, as the case for non-positive functions follows analogously.

\begin{lemma}\label{lem:selection}
    Let $\phi \in W^{1,2}_{\rm loc}(\mathbb R^{D-1}\times S^1)$ be a function such that $\mathcal T[\phi]$ and $\mathcal K[\phi]+\mathcal V[\phi]<0$ are finite. Consider a decomposition $\phi = \phi_1 + \phi_2$, where $\phi_1$ and $\phi_2$ have disjoint supports. If $\mathcal V[\phi]+\mathcal K[\phi]<0$, then either $\mathcal R[\phi_1]\leq\mathcal R[\phi]$ or $\mathcal R[\phi_2]\leq\mathcal R[\phi]$. Moreover, equality holds if and only if $\phi_2=0$ $\mathcal H^D$-a.e.~in the former case, or $\phi_1=0$ $\mathcal H^D$-a.e.~in the latter.
\end{lemma}

\begin{proof}
    Let $U_1 = \supp(\phi_1)$ and $U_2 = \supp(\phi_2)$ be the supports of $\phi_1$ and $\phi_2$, respectively.
    By construction, we have
    \begin{align}
        \nabla_x\phi_1&=(\nabla_x\phi)\chi_{U_1},\quad
        \nabla_x\phi_2=(\nabla_x\phi)\chi_{U_2},\\
        \partial_\tau\phi_1&=(\partial_\tau\phi)\chi_{U_1},\quad
        \partial_\tau\phi_2=(\partial_\tau\phi)\chi_{U_2},
    \end{align}
    $\mathcal H^D$-a.e., where $\chi_{U_i}$ denotes the indicator function.
    From these relations and the assumption $V(0)=0$, we obtain
    \begin{align}
        \mathcal T[\phi]&=\mathcal T[\phi_1]+\mathcal T[\phi_2],\\
        \mathcal K[\phi]&=\mathcal K[\phi_1]+\mathcal K[\phi_2],\\
        \mathcal V[\phi]&=\mathcal V[\phi_1]+\mathcal V[\phi_2].
    \end{align}
    Consequently, $\mathcal R[\phi]$ can be expressed as
    \begin{equation}
        \mathcal R[\phi]=\frac{(\mathcal T[\phi_1]+\mathcal T[\phi_2])^{\frac{D-1}{D-3}}}{-(\mathcal V[\phi_1]+\mathcal K[\phi_1])-(\mathcal V[\phi_2]+\mathcal K[\phi_2])}.
    \end{equation}

    Since the denominator must be positive, there are three possible cases:
    \begin{enumerate}
        \item $\mathcal V[\phi_1]+\mathcal K[\phi_1]<0$ and $\mathcal V[\phi_2]+\mathcal K[\phi_2]\geq0$,
        \item $\mathcal V[\phi_1]+\mathcal K[\phi_1]\geq0$ and $\mathcal V[\phi_2]+\mathcal K[\phi_2]<0$,
        \item $\mathcal V[\phi_1]+\mathcal K[\phi_1]<0$ and $\mathcal V[\phi_2]+\mathcal K[\phi_2]<0$.
    \end{enumerate}
    For each case, we find the following inequities\footnote{Case 1 and case 2 are trivial. For case 3, we used
    \begin{equation}
        \ab(a+b)^p\geq a^p+b^p,\label{eq:minkowski}
    \end{equation}
    for $a,b>0$ and $p>1$, which follows from Minkowski's inequality and the equality holds if and only if $a=0$ or $b=0$.
    We also used the mediant inequality,
    \begin{equation}
        \frac{a+b}{c+d}=\frac{c}{c+d}\frac{a}{c}+\frac{d}{c+d}\frac{b}{d}\geq\min\ab(\frac{a}{c},\frac{b}{d}),\label{eq:mediant}
    \end{equation}
    for $a,b,c,d>0$.}:
    \begin{enumerate}
        \item $\mathcal R[\phi_1]\leq \mathcal R[\phi]$,
        \item $\mathcal R[\phi_2]\leq \mathcal R[\phi]$,
        \item $\min(\mathcal R[\phi_1],\mathcal R[\phi_2])< \mathcal R[\phi]$.
    \end{enumerate}
    In case 1, equality holds if and only if $\phi_2=0$ $\mathcal H^D$-a.e., while in case 2, it holds if and only if $\phi_1=0$ $\mathcal H^D$-a.e.
\end{proof}
\begin{lemma}\label{lem:non-negativity}
    Either there exists a minimizing sequence such that $\phi_n(z)\geq0$ for all $n$ and all $z\in \mathbb R^{D-1}\times S^1$, or there exists a minimizing sequence such that $\phi_n(z)\leq0$ for all $n$ and all $z\in \mathbb R^{D-1}\times S^1$.
\end{lemma}
\begin{proof}
    Any function $\phi\in W^{1,2}_{\rm loc}(\mathbb R^{D-1}\times S^1)$ can be written as the sum of its positive and negative parts,
    \begin{equation}
        \phi(z)=\phi_+(z)+\phi_-(z),~\phi_+(z)=\max\ab\{\phi(z),0\},~\phi_-(z)=\min\ab\{\phi(z),0\}.
    \end{equation}
    Since $\phi_+$ and $\phi_-$ have disjoint supports, from Lemma \ref{lem:selection}, we have either $\mathcal R[\phi_+]\leq\mathcal R[\phi]$ or $\mathcal R[\phi_-]\leq\mathcal R[\phi]$.

    By selecting the components that yield smaller or equal values of $\mathcal R$, we can construct two sequences: $\{\phi_n : \mathcal R[\phi_{n+}] \leq \mathcal R[\phi_n]\}$ and $\{\phi_n : \mathcal R[\phi_{n-}] \leq \mathcal R[\phi_n]\}$. Since $\{\phi_n\}$ is a minimizing sequence, at least one of these sequences has to be a minimizing sequence. Note that both can be minimizing sequences if they yield the same value of $\mathcal R$ as $n\to\infty$. Finally, we rescale the elements of the new minimizing sequence, $\phi_{n\pm}$, such that $\mathcal V[\phi_{n\pm,\lambda}] = \mathcal V[\phi_n]$, thereby satisfying the conditions of the reduced problem and proving the lemma.
\end{proof}

\begin{lemma}\label{lem:non-negative-sol}
    If $\phi$ is a solution of the reduced problem, $\phi$ is either $\mathcal H^D$-a.e.~non-negative or $\mathcal H^D$-a.e.~non-positive.
\end{lemma}
\begin{proof}
    Suppose, for the sake of contradiction, that $\phi$ has both positive and negative parts of non-zero measure. From Lemma \ref{lem:selection}, this would imply that either $\mathcal{R}[\phi_+] < \mathcal{R}[\phi]$ or $\mathcal{R}[\phi_-] < \mathcal{R}[\phi]$. This contradicts the assumption that $\phi$ is a solution to the reduced problem, which by definition minimizes $\mathcal{R}$.
\end{proof}
\subsection{Connected support}
A key distinction from the zero-temperature analysis is the need to address configurations with disconnected support. Unlike the full symmetrization employed in the CGM proof, Steiner symmetrization does not guarantee that the support of a function becomes connected. Physically, a disconnected support corresponds to a multi-instanton configuration, and we will show that the value of the functional $\mathcal{R}$ for such a configuration is strictly greater than that of one of its single-instanton components. This step is crucial, as functions with disconnected support are known counterexamples to the categorization of equality cases \cite{CIANCHI2006673}.
In this subsection, we demonstrate that any minimizing sequence can be refined to consist of functions whose support is connected in a specific measure-theoretical sense, which we define below. Subsequently, we will assume, without loss of generality, that the support of each function in our minimizing sequence is connected.
\begin{definition}
    Let $\Omega\subset R^n$ or $\Omega\subset R^{n-1}\times S^1$, and let $u:\Omega\to[0,\infty)$. For $\epsilon>0$, we define the measure-theoretical interior of the strict superlevel set by
    \begin{equation}
        \{u>\epsilon\}^{(1)}=\bigcup\{E\subset \Omega:u(z)>\epsilon~\text{\rm for}~\mathcal H^n\text{\rm -a.e.}~z\in E,~E~\text{\rm open}\},
    \end{equation}
    and also define the measure-theoretical support\footnote{This definition is chosen to be suitable for our analysis, which relies on the connectedness of the support, and it slightly differs from the standard definition.} of $u$ by
    \begin{equation}
        \{u>0^+\}^{(1)}=\bigcup_{k=1}^\infty\ab\{u>\frac{1}{k}\}^{(1)}.
    \end{equation}
\end{definition}

\begin{lemma}
    There exists a minimizing sequence such that $\{\phi_n>0^+\}^{(1)}$ is connected for all $n$.
\end{lemma}
\begin{proof}
    We decompose the support as $\{\phi_n>0^+\}^{(1)} = \bigcup_{i\geq1} U_n^i$, where each $U_n^i$ is selected so that $U_n^i$ constitutes a non-empty, maximal connected subset of $\{\phi_n>0^+\}^{(1)}$.

    First, consider the case where the number of connected components is finite. By iteratively applying Lemma \ref{lem:selection} to the functions $\phi_n\chi_{U_n^i}$ and $\Phi_n^{\{l>i\}}=\sum_{l> i}\phi_n\chi_{U_n^l}$ for $i=1,2,\dots$, we can identify a component $j$ such that
    \begin{equation}
        \mathcal V[\phi_n\chi_{U_n^j}]+\mathcal K[\phi_n\chi_{U_n^j}]<0,\label{eq:con_cond1}
    \end{equation}
    and
    \begin{equation}
        \mathcal R[\phi_n\chi_{U_n^j}]<\mathcal R\ab[\Phi_n^{\{l>j-1\}}]<\cdots<\mathcal R\ab[\Phi_n^{\{l>2\}}]<\mathcal R\ab[\Phi_n^{\{l>1\}}]<\mathcal R[\phi_n].\label{eq:con_cond2}
    \end{equation}
    This demonstrates that, for a finite number of disconnected components, we can always identify a single component that yields a strictly smaller value of $\mathcal R$.

    Next, we address the case where the number of components is infinite. We repeat the same iterative procedure for a finite number of steps, $k\geq1$. If at any step we identify a component $j$ satisfying Eqs.~\eqref{eq:con_cond1} and \eqref{eq:con_cond2}, the argument is complete. Otherwise, the resultant $\mathcal R\ab[\Phi_n^{\{l>k\}}]$ is strictly smaller than $\mathcal R[\phi_n]$, which implies that there exists a positive constant $C<1$ such that
    \begin{equation}
        \mathcal R\ab[\Phi_n^{\{l>k\}}]<C\mathcal R[\phi_n].\label{eq:halfR}
    \end{equation}
    Since $\mathcal V[\Phi_n^{\{l>k\}}] + \mathcal K[\Phi_n^{\{l>k\}}]$ is finite, we can partition the remaining infinite set of components $\{l \in \mathbb{N} : l > k\}$ into a finite subset $F$ and its complement $I = \{l \in \mathbb{N} : l > k\} \setminus F$ such that
    \begin{equation}
        \ab(\frac{1}{C}-1)\ab(\mathcal V\ab[\Phi_n^F]+\mathcal K\ab[\Phi_n^F])<\mathcal V\ab[\Phi_n^I]+\mathcal K\ab[\Phi_n^I].\label{eq:fin_vs_inf}
    \end{equation}
    Here, $\Phi_n^F=\sum_{i\in F}\phi_n\chi_{U_n^i}$ and $\Phi_n^I=\sum_{i\in I}\phi_n\chi_{U_n^i}$.
    We will now demonstrate that
    \begin{equation}
        \mathcal R\ab[\Phi_n^F]<\mathcal R[\phi_n].\label{eq:tobeshown}
    \end{equation}

    Applying Lemma \ref{lem:selection} to the decomposition $\Phi_n^{\{l>k\}} = \Phi_n^F + \Phi_n^I$, we have two possibilities.
    If $\mathcal R[\Phi_n^F] \leq \mathcal R[\Phi_n^{\{l>k\}}]$, then Eq.~\eqref{eq:tobeshown} follows directly from Eq.~\eqref{eq:halfR}.
    Alternatively, if $\mathcal R[\Phi_n^I] \leq \mathcal R[\Phi_n^{\{l>k\}}]$, we can use Eq.~\eqref{eq:fin_vs_inf} to improve the upper bound on $\mathcal R[\Phi_n^F]$:
    \begin{align}
        \mathcal R\ab[\Phi_n^{\{l>k\}}]&\geq \frac{(\mathcal T[\Phi_n^F])^{\frac{D-1}{D-3}}+(\mathcal T[\Phi_n^I])^{\frac{D-1}{D-3}}}{-(\mathcal V[\Phi_n^F]+\mathcal K[\Phi_n^F])-(\mathcal V[\Phi_n^I]+\mathcal K[\Phi_n^I])}>C\mathcal R[\Phi_n^F].
    \end{align}
    Combining this with Eq.~\eqref{eq:halfR} again implies Eq.~\eqref{eq:tobeshown}.
    Thus, in either case, we have identified a function $\Phi_n^F$, corresponding to a finite number of components, that has a smaller value of $\mathcal R$ than the original function $\phi_n$. By applying the argument for the finite case to $\Phi_n^F$, we can isolate a single component $j$ that satisfies Eqs.~\eqref{eq:con_cond1} and \eqref{eq:con_cond2}.

    This process yields a new sequence of functions, each associated with a single connected component and a strictly smaller value of $\mathcal{R}$. To ensure this new sequence adheres to the constraints of the reduced problem, we rescale each function $\phi_n \chi_{U_n^j}$ to a new function $\check{\phi}_n$ such that $\mathcal{V}[\check{\phi}_n] = \mathcal{V}[\phi_n]$. The resulting sequence $\{\check{\phi}_n\}$ is a minimizing sequence where each element has a connected support, as required.
\end{proof}

\begin{lemma}
    If $\phi$ is a solution of the reduced problem, $\{\phi>0^+\}^{(1)}$ is connected.
\end{lemma}
\begin{proof}
    The proof proceeds by contradiction, analogous to that of Lemma \ref{lem:non-negative-sol}. If the support were disconnected, one could select a single connected component that yields a strictly smaller value of $\mathcal{R}$, which would contradict the assumption that $\phi$ is a minimizer.
\end{proof}

\subsection{Spherical symmetry and monotonicity}
We now turn our attention to establishing the symmetry of the solution. The subsequent subsections will summarize the relevant results from Capriani \cite{Capriani2011TheSR}, extend these theorems to functions vanishing at spatial infinity, and then apply them to our reduced problem. For the remainder of this section, we will assume, without loss of generality, that $\phi_n$ is spherically symmetric and monotonic with respect to $x$ for all $n$.

\subsubsection{Steiner symmetrization}
This subsection summarizes key definitions and theorems from Capriani's work on Steiner rearrangement \cite{Capriani2011TheSR}, which are essential for our analysis. We adopt the notation where $z=(x,y)$, with $x \in \mathbb{R}^{k}$ and $y \in \mathbb{R}^{n-k}$, and we are particularly interested in the case where $n=D$ and $k=D-1$. Accordingly, the derivatives of function $u$ are denoted by $\nabla u=(\nabla_x u,\nabla_y u)$.
Note that $x$ and $y$ have been swapped to adapt to our notation.

The volume of the $k$ dimensional unit ball is denoted by $\omega_k^{\mathcal H}=\mathcal H^k(\{x\in \mathbb R^k: |x|<1\})$.
For a measurable set $E\subset \mathbb{R}^{k}\times\mathbb{R}^{n-k}$, its section at $y$ is denoted by $E_y = \{x \in \mathbb{R}^k : (x,y) \in E\}$, its projection onto $\mathbb{R}^{n-k}$ is $\pi_{n-k}(E) = \{y \in \mathbb{R}^{n-k} : (x,y) \in E\}$, and the essential projection is $\pi^+_{n-k}(E) = \{y \in \mathbb{R}^{n-k} : (x,y) \in E,\mathcal H^k(E_y)>0\}$.
The distribution function of $u(\cdot,y)$ is $\lambda_u(y,t) = \mathcal{H}^k(\{x \in \mathbb{R}^k : u_0(x,y) > t\})$, $u_0$ is the extension of $u$ by $0$ outside $E$, and $M_u(y) = \inf\{t>0 : \lambda_u(y,t)=0\}$ represents the essential supremum of $u$ at $y$.
In the following, $u$ and $u_0$ are used interchangeably when there is no risk of confusion.

The relevant theorems apply to functions in the space $W^{1,1}_{0,x}(\Omega)$, which consists of functions that are locally Sobolev in the $y$ variable and globally Sobolev with zero trace in the $x$ variable. Formally, for an open set $\Omega\subset\mathbb R^n$, this space is defined as:
\begin{equation}
    W^{1,p}_{0,x}(\Omega)=\{u:\Omega\to\mathbb R:u_0\in W^{1,p}(\mathbb R^k\times\omega),~\forall\omega\Subset\pi_{n-k}(\Omega),~\omega~{\rm open}\},
\end{equation}
where $\Subset$ denotes compact inclusion, {\it i.e.}, for two open sets $\omega$ and $\Omega$, $\omega\Subset\Omega$ indicates $\bar\omega\subset\Omega$ and $\bar\omega$ is compact.

For a set $E$ of finite perimeter in an open set $\Omega \subset \mathbb{R}^n$, the generalized inner normal $\nu^E$ is defined as the derivative of Radon measure $D\chi_E$ with respect to $|D\chi_E|$, and its $x$-component is denoted by $\nu^E_x$. The reduced boundary $\partial^*E$ comprises all points $z \in \Omega$ where $\nu^E(z)$ exists and satisfies $|\nu^E(z)|=1$.

\begin{definition}
    The Steiner rearrangement of a measurable set $E \subset \mathbb{R}^k \times \mathbb{R}^{n-k}$ with respect to the $x$ variable, denoted by $E^\sigma$, is defined as
    \begin{align}
        E^\sigma&=\ab\{(x,y)\in \mathbb R^k\times \mathbb R^{n-k}:y\in\pi_{n-k}^+(E),~|x|^k\leq\frac{\mathcal H^k(E_y)}{\omega_k^{\mathcal H}}\}.
    \end{align}
\end{definition}
\begin{definition}
    Let $u$ be a non-negative measurable function on a set $E \subset \mathbb{R}^k \times \mathbb{R}^{n-k}$ such that its superlevel sets have finite measure on each slice, {\it i.e.}, $\mathcal H^k(\{x\in E_y:u(x,y)>t\})<\infty$ for all $t>0$ and for $\mathcal H^{n-k}$-a.e.~$y\in\pi_{n-k}^+(E)$.
    The Steiner rearrangement $u^\sigma$ of $u$ is defined as
    \begin{align}
        u^\sigma(x,y)&=\inf\ab\{t>0:\lambda_u(y,t)\leq \omega_k^{\mathcal H}|x|^k\}.
    \end{align}
\end{definition}
\begin{definition}
    Let $f:\mathbb R^k\times \mathbb R^{n-k}\to[0,\infty)$. Then, $f$ is said to be {\it radially symmetric} if there exists a function $\tilde f:\mathbb R^{n-k+1}\to[0,\infty)$ such that
    \begin{equation}
        f(x,y)=\tilde f(|x|,y),
    \end{equation}
    for all $(x,y)\in\mathbb R^k\times \mathbb R^{n-k}$.
\end{definition}
\begin{theorem}[Theorem 2.1 of \cite{Capriani2011TheSR}]\label{thm:steiner}
    Let $f:\mathbb R^k\times \mathbb R^{n-k}\to[0,\infty)$ be a non-negative convex function that vanishes at $0$ and is radially symmetric.
    Let $\Omega\subset\mathcal R^n$ be an open set and $u\in W_{0,x}^{1,1}(\Omega)$ be a non-negative function. Then
    \begin{equation}
        \int_{\Omega^\sigma}f(\nabla u^\sigma)\odif[order=n]{z}\leq\int_{\Omega}f(\nabla u)\odif[order=n]{z}.
    \end{equation}
\end{theorem}
\begin{theorem}[Theorem 2.3 of \cite{Capriani2011TheSR}]\label{thm:equality_cases}
    Let $f:\mathbb R^k\times \mathbb R^{n-k}\to[0,\infty)$ be a non-negative strictly convex function that vanishes at $0$ and is radially symmetric. Let $\Omega\subset\mathcal R^n$ be an open set satisfying
    \begin{itemize}
        \item $\pi_{n-k(\Omega)}$ is connected;
        \item $\Omega$ is bounded in the $x$ directions;
        \item $\Omega$ is of finite perimeter inside $\mathbb R^k\times\pi_{n-k}(\Omega)$;
        \item $\mathcal H^{n-1}(\{(x,y)\in\partial^*\Omega:\nu_x^\Omega=0\}\cap\{\mathbb R^k\times\pi_{n-k}(\Omega)\})=0$.
    \end{itemize}
    Let $u\in W_{0,x}^{1,1}(\Omega)$ be a non-negative function. If
    \begin{equation}
        \int_{\Omega^\sigma}f(\nabla u^\sigma)\odif[order=n]{z}=\int_{\Omega}f(\nabla u)\odif[order=n]{z}<\infty,\label{eq:equality_case}
    \end{equation}
    and $u$ satisfies
    \begin{align}
        \mathcal H^n(&\{(x,y)\in\Omega:\nabla_xu(x,y)=0\}\nonumber\\
        &\cap\{(x,y)\in\Omega:{\rm either}~M_u(y)=0~{\rm or}~u(x,y)<M_u(y)\})=0,
    \end{align}
    then $u^\sigma$ is equivalent to $u$ up to a translation in the $x$-plane.
\end{theorem}
\subsubsection{Generalization to functions vanishing at infinity}
The original theorems are formulated for functions in $W^{1,1}_{0,x}(\Omega)$, which is a subset of $L^1(\Omega)$. However, the functions relevant to our physical problem are not necessarily in $L^1(\mathbb{R}^{D-1} \times S^1)$ and are required only to vanish at infinity in the $x$ directions. It is therefore necessary to extend these theorems to accommodate such functions. To this end, we introduce a localized version of $W^{1,p}_{0,x}(\Omega)$ as
\begin{equation}
    W^{1,p}_{0,x,{\rm loc}}(\Omega)=\{u\in W_{\rm loc}^{1,p}(\Omega):\forall \varphi\in C^\infty_c(\Omega),\varphi u\in W^{1,p}_{0,x}(\Omega)\},
\end{equation}
and we say that a function $u\in W^{1,p}_{0,x,{\rm loc}}(\Omega)$ vanishes at infinity in the $x$ directions if it satisfies
\begin{equation}
    \forall \epsilon>0,~\exists~{\rm compact}~K\subset \mathbb R^k~\text{\rm s.t.}~|u_0(z)|<\epsilon~\text{\rm for }\mathcal H^n\text{\rm -a.e. } z\not\in (K\times\pi_{n-k}(\Omega)).
\end{equation}

\begin{corollary}\label{cor:steiner}
    Let $f:\mathbb R^k\times \mathbb R^{n-k}\to[0,\infty)$ be a non-negative convex function that vanishes at $0$ and is radially symmetric.
    Let $\Omega\subset\mathcal R^n$ be an open set and $u\in W^{1,p}_{0,x,\rm loc}(\Omega)$ be a non-negative function vanishing at infinity in the $x$ directions.
    Then
    \begin{equation}
        \int_{\Omega^\sigma}f(\nabla u^\sigma)\odif[order=n]{z}\leq\int_{\Omega}f(\nabla u)\odif[order=n]{z}.\label{eq:cor_steiner}
    \end{equation}
\end{corollary}
\begin{proof}
    The proof adapts the argument from Remark 4.5 of \cite{Capriani2011TheSR}.
    For every $\epsilon>0$, we define
    \begin{equation}
        u_\epsilon=\max(u-\epsilon,0).
    \end{equation}
    Then, there exists compact $K\subset \mathbb R^k$ such that $\esssupp(u_\epsilon)\subset K\times\pi_{n-k}(\Omega)$ and thus $u_\epsilon\in W_{0,x}^{1,1}(\Omega)$. We have $(u_\epsilon)^\sigma=(u^\sigma)_\epsilon$ and $\nabla u_\epsilon=(\nabla u)\chi_{\{z:u(z)>\epsilon\}}$ $\mathcal H^n$-a.e.~in $\mathbb R^n$. By Eq.~(4.13) of \cite{Capriani2011TheSR} and the monotone convergence theorem, we obtain
    \begin{align}
        \int_{\mathbb R^k\times B} f(\nabla u^\sigma)\odif[order=n]{z}&=\lim_{\epsilon\to0^+}\int_{\mathbb R^k\times B} f(\nabla (u^\sigma)_\epsilon)\odif[order=n]{z}\nonumber\\
        &=\lim_{\epsilon\to0^+}\int_{\mathbb R^k\times B} f(\nabla (u_\epsilon)^\sigma)\odif[order=n]{z}\nonumber\\
        &\leq\lim_{\epsilon\to0^+}\int_{\mathbb R^k\times B} f(\nabla u_\epsilon)\odif[order=n]{z}\nonumber\\
        &=\int_{\mathbb R^k\times B} f(\nabla u)\odif[order=n]{z},
    \end{align}
    for every Borel set $B\subset \pi_{n-k}(\Omega)$. This is a stronger statement than Eq.~\eqref{eq:cor_steiner}.
\end{proof}

\begin{corollary}\label{cor:equality_cases}
    Let $\Omega\subset\mathcal R^n$ be an open set, and $f:\mathbb R^k\times \mathbb R^{n-k}\to[0,\infty)$ be a non-negative strictly convex function that vanishes at $0$ and is radially symmetric.
    Let $u\in W^{1,p}_{0,x,\rm loc}(\Omega)$ be a non-negative function vanishing at infinity in the $x$ directions.
    If $u$ satisfies
    \begin{align}
        \mathcal H^n(&\{(x,y)\in \{u>0^+\}^{(1)}:\nabla_x u(x,y)=0\}\nonumber\\
        &\cap\{(x,y)\in \{u>0^+\}^{(1)}:u(x,y)<M_u(y)\})=0,\label{eq:eqality_cond1}
    \end{align}
    and
    \begin{equation}
        \int_{\Omega^\sigma} f(\nabla u^\sigma)\odif[order=n]{z}=\int_\Omega f(\nabla u)\odif[order=n]{z}<\infty,
    \end{equation}
    then, for each connected domain of $\pi_{n-k}(\{u>0^+\}^{(1)})$,
    $u$ is equivalent to $u^\sigma$ up to a translation in the $x$-plane.
\end{corollary}
\begin{proof}
    For any $\epsilon > 0$, the set $\{u > \epsilon\}^{(1)}$ is open by definition and bounded in the $x$ directions, as there exists a compact set $K\subset\mathbb R^k$ such that $\{u>\epsilon\}^{(1)}\subset K\times\pi_{n-k}(\Omega)$. Furthermore, since $u$ is a Sobolev function, the coarea formula guarantees that $\{u > \epsilon\}^{(1)}$ has finite perimeter for a.e.~$\epsilon > 0$.

    We decompose the set $\{u>\epsilon\}^{(1)}$ as $\{u>\epsilon\}^{(1)} = \bigcup_i U_\epsilon^i$, where each $U_\epsilon^i$ is selected so that $\pi_{n-k}(U_\epsilon^i)$ constitutes a non-empty, maximal connected subset of $\pi_{n-k}(\{u>\epsilon\}^{(1)})$. Since $U_\epsilon^i$ is a non-empty open set, it follows that $M_u(y)>\epsilon$ for $\mathcal H^{n-k}$-a.e.~$y\in\pi_{n-k}(U_\epsilon^i)$.

    Define
    \begin{equation}
        u_\epsilon^i = \max(u - \epsilon, 0) \chi_{U_\epsilon^i}.
    \end{equation}
    Then, the interior approximate limit of $u_\epsilon^i$ at the boundary of $U_\epsilon^i$ is zero, implying that $u_\epsilon^i$ has zero trace in the $x$ directions. Consequently, $u_\epsilon^i\in W_{0,x}^{1,1}(U_\epsilon^i)$.

    Since $\nabla_xu_\epsilon^i(z)=\nabla_xu(z)$ for $\mathcal H^n$-a.e.~$z\in U_\epsilon^i$, Eq.~\eqref{eq:eqality_cond1} implies
    \begin{align}
        \mathcal H^n(&\{(x,y)\in U_\epsilon^i:\nabla_xu_\epsilon^i(x,y)=0\}\nonumber\\
        &\cap\{(x,y)\in U_\epsilon^i:{\rm either}~M_{u_\epsilon^i}(y)=0~{\rm or}~u_\epsilon^i(x,y)<M_{u_\epsilon^i}(y)\})=0.
    \end{align}
    The first condition inside the left set definition is redundant because $M_{u_\epsilon^i}(y)>0$ holds for $\mathcal H^{n-k}$-a.e.~$y\in \pi_{n-k}(U_\epsilon^i)$ as mentioned above.

    Furthermore, to apply Theorem \ref{thm:equality_cases}, we must verify that the condition on the normal vector holds on the boundary. By definition, the interior approximate limit of $u$ at the boundary of $U_\epsilon^i$ is $\epsilon$. Since $M_u(y) > \epsilon$ for a.e.~$y \in \pi_{n-k}(U_\epsilon^i)$, it follows that $u(x,y) < M_u(y)$ for $\mathcal{H}^{n-1}$-a.e.~$(x,y)\in\partial^* U_\epsilon^i \cap (\mathbb{R}^k \times \pi_{n-k}(U_\epsilon^i))$. The condition Eq.~\eqref{eq:eqality_cond1} then implies that $\nabla_x u(z) \neq 0$ for $\mathcal{H}^{n-1}$-a.e.~$z$ on this part of the boundary and for $\mathcal H$-a.e.~$\epsilon>0$. According to Theorem 3.3 of \cite{Capriani2011TheSR}, the normal vector component $\nu_x^{U_\epsilon^i}$ is proportional to $\nabla_x u$ for $\mathcal{H}^{n-1}$-a.e.~$z$ on the boundary and for $\mathcal{H}$-a.e.~$\epsilon>0$. Consequently, for $\mathcal H$-a.e.~$\epsilon>0$, we have
    \begin{equation}
        \mathcal H^{n-1}(\{(x,y)\in\partial^*U_\epsilon^i:\nu_x^{U_\epsilon^i}=0\}\cap\{\mathbb R^k\times\pi_{n-k}(U_\epsilon^i)\})=0.
    \end{equation}

    The preceding arguments establish that for $\mathcal{H}$-a.e.~$\epsilon > 0$, the function $u_\epsilon^i$ and the set $U_\epsilon^i$ satisfy all the hypotheses of Theorem \ref{thm:equality_cases}. Consequently, for each component $U_\epsilon^i$, the function $u_\epsilon^i$ must be equivalent to its Steiner rearrangement $(u_\epsilon^i)^\sigma$ up to a translation in the $x$-plane. Since this holds for $\mathcal{H}$-a.e.~$\epsilon > 0$, the equivalence extends from the truncated functions $u_\epsilon^i$ to the original function $u$: On each connected component of $\pi_{n-k}(\{u>0^+\}^{(1)})$, $u$ is equivalent to $u^\sigma$ up to a translation in the $x$-plane.
    \end{proof}

    \subsubsection{Application to the reduced problem}

    To apply Corollaries \ref{cor:steiner} and \ref{cor:equality_cases} to the reduced problem, we identify $\phi$ with a function on $\mathbb{R}^{D-1} \times (0, \beta)$ by removing a point from $S^1$. Specifically, if the projection $\pi_{S^1}(\{\phi>0^+\}^{(1)})$ does not cover the entire circle, we choose this point from its complement; otherwise, any point suffices. We will henceforth use this identification implicitly. Under this identification, $\phi$ belongs to $W_{\rm loc}^{1,2}(\mathbb{R}^{D-1}\times (0,\beta))$ and vanishes at spatial infinity. We note that the zero trace condition in the spatial directions is inherent. The Steiner symmetrization, denoted by $\cdot^\sigma$, is performed with respect to the spatial coordinates $x \in \mathbb{R}^{D-1}$. This operation leaves the domain invariant and preserves the periodicity in the $\tau$ direction.

    \begin{lemma}
        Let $\phi \in W^{1,2}_{\rm loc}(\mathbb{R}^{D-1}\times S^1)$ be a non-negative function vanishing at spatial infinity with $\mathcal{T}[\phi] < \infty$ and $\mathcal{K}[\phi] < \infty$. Then $\phi^\sigma$ belongs to $W^{1,2}_{\rm loc}(\mathbb{R}^{D-1}\times S^1)$.
    \end{lemma}

    \begin{proof}
        For any $\epsilon > 0$, define the truncated function
        \begin{equation}
            \phi_\epsilon = \max(\phi - \epsilon, 0).
        \end{equation}
    Given that $\esssupp(\phi_\epsilon)\subset K\times S^1$ for some compact $K\subset \mathbb R^{D-1}$, it follows that $\phi_\epsilon\in L^2(\mathbb R^{D-1}\times S^1)$. Consequently, its Steiner symmetrization $(\phi_\epsilon)^\sigma$ also belongs to $L^2(\mathbb R^{D-1}\times S^1)$. Applying the identity $(\phi_\epsilon)^\sigma=(\phi^\sigma)_\epsilon$ and the Minkowski inequality, for an arbitrary open set $U\Subset\mathbb R^{D-1}\times S^1$, we obtain:
    \begin{align}
        \int_U|\phi^\sigma|^2\odif[order=D]{z}&\leq\int_U|(\phi^\sigma)_
        \epsilon+\epsilon|^2\odif[order=D]{z}\nonumber\\
        &\leq\ab[\ab(\int_U|(\phi^\sigma)_
        \epsilon|^2\odif[order=D]{z})^{1/2}+\epsilon\ab(\int_U\odif[order=D]{z})^{1/2}]^2<\infty.
    \end{align}
    Furthermore, by applying Corollary \ref{cor:steiner} with $f(\nabla\phi)=|\nabla\phi|^2$, we find that
    \begin{equation}
        \int_U|\nabla \phi^\sigma|^2\odif[order=D]{z}\leq\int|\nabla \phi^\sigma|^2\odif[order=D]{z}\leq\int|\nabla \phi|^2\odif[order=D]{z}<\infty.
    \end{equation}
    Thus, $\phi^\sigma$ belongs to $W^{1,2}_{\rm loc}(\mathbb R^{D-1}\times S^1)$.
\end{proof}

\begin{lemma}\label{lem:symmetry}
    There exists a minimizing sequence such that $\phi_n$ is $\mathcal H^D$-a.e.~spherically symmetric and monotonic with respect to $x$ for all $n$.
\end{lemma}
\begin{proof}
    As $\phi_n \in W^{1,2}_{\rm loc}(\mathbb R^{D-1}\times S^1)$ is a non-negative function vanishing at spatial infinity, we can apply Corollary \ref{cor:steiner}. By choosing the convex functions $f(\nabla\phi)=|\nabla_x\phi|^2$ and $f(\nabla\phi)=|\partial_\tau\phi|^2$, we obtain the following inequalities for the kinetic and gradient terms:
    \begin{align}
        \mathcal T[\phi_n^\sigma]&\leq\mathcal T[\phi_n],\\
        \mathcal K[\phi_n^\sigma]&\leq\mathcal K[\phi_n].
    \end{align}
    Furthermore, since Steiner symmetrization preserves the measure of level sets, the potential term remains unchanged: $\mathcal V[\phi_n^\sigma]=\mathcal V[\phi_n]$.
    
    Combining these results, it follows that the functional $\mathcal R$ does not increase under symmetrization: $\mathcal R[\phi_n^\sigma]\leq\mathcal R[\phi_n]$. Therefore, we can replace the original minimizing sequence $\{\phi_n\}$ with the sequence of its Steiner symmetrizations, $\{\phi_n^\sigma\}$. By construction, each function in this new sequence is spherically symmetric and monotonic with respect to $x$, which proves the lemma.
\end{proof}

\begin{lemma}
    Let $\phi$ be a solution of the reduced problem.
    If $\phi$ satisfies
    \begin{align}
        \mathcal H^D(&\{(x,\tau)\in \{\phi>0^+\}^{(1)}:\nabla_x\phi(x,\tau)=0\}\nonumber\\
        &\cap\{(x,\tau)\in \{\phi>0^+\}^{(1)}:\phi(x,\tau)<M(\tau)\})=0,\label{eq:no-plateau}
    \end{align}
    then, $\phi$ is $\mathcal H^D$-a.e.~spherically symmetric and monotonic with respect to $x$.
\end{lemma}
\begin{proof}
    Since $\phi$ is a solution to the reduced problem, it is non-negative, its support $\{\phi>0^+\}^{(1)}$ is connected, and it satisfies $\mathcal R[\phi^\sigma]=\mathcal R[\phi]$.

    The equality $\mathcal R[\phi^\sigma]=\mathcal R[\phi]$, combined with the inequalities $\mathcal T[\phi^\sigma]\leq\mathcal T[\phi]$ and $\mathcal K[\phi^\sigma]\leq\mathcal K[\phi]$ and the fact that $\mathcal V[\phi^\sigma]=\mathcal V[\phi]$, implies that both inequalities must be saturated. Thus, we must have $\mathcal T[\phi^\sigma]=\mathcal T[\phi]$ and $\mathcal K[\phi^\sigma]=\mathcal K[\phi]$.
    
    Let $f(\nabla\phi)=|\nabla\phi|^2$, then $f$ is a non-negative strictly convex function, vanishing at zero. We have
    \begin{align}
        \int\odif[order=D]{z}f(\nabla\phi^\sigma)&=\int\odif[order=D]{z}|\nabla_x\phi^\sigma|^2+\int\odif[order=D]{z}|\partial_\tau\phi^\sigma|^2\nonumber\\
        &=\int\odif[order=D]{z}|\nabla_x\phi|^2+\int\odif[order=D]{z}|\partial_\tau\phi|^2\nonumber\\
        &=\int\odif[order=D]{z}f(\nabla\phi)<\infty.
    \end{align}

    The preceding analysis confirms that the function $\phi$, the set $\{\phi>0^+\}^{(1)}$, and the convex function $f$ satisfy all the hypotheses of Corollary \ref{cor:equality_cases}. This implies that $\phi$ is equivalent to its Steiner rearrangement $\phi^\sigma$ up to a translation in the $x$-plane. This establishes that $\phi$ is $\mathcal{H}^D$-a.e.~spherically symmetric and monotonic with respect to $x$.
\end{proof}
\subsection{Existence of a solution to the reduced problem}
Finally, we establish the existence of a solution to the reduced problem. This step is crucial, as convergence of the minimizing sequence to the trivial configuration would render our preceding arguments a mere demonstration of the trivial configuration's spherical symmetry.

For the subsequent analysis, we introduce the following notation. Let $B(p,R)$ denote the $(D-1)$-dimensional open ball of radius $R>0$ centered at $p\in \mathbb R^{D-1}$. We define the integral domains $\Omega_R=B(0,R)\times S^1$ and its complement, $\Omega_R^c=B^c(0,R)\times S^1$. The volume of a $k$-dimensional unit ball is denoted by $\omega_k=\int_{|x|<1}\odif[order=k]{x}$. For a minimizing sequence $\{\phi_n\}$, the potential term is fixed at $\mathcal V[\phi_n] = \mathcal V_0 < 0$. Without loss of generality, we assume the sequence is bounded from above, {\it i.e.}, $\sup_n\mathcal R[\phi_n]=\mathcal R_0<\infty$.

\begin{lemma}\label{lem:uniform}
    Let $\{\phi_n\}$ be a minimizing sequence. Then, the gradient term $\mathcal T[\phi_n]$ and the kinetic term $\mathcal K[\phi_n]$ are uniformly bounded from above.
\end{lemma}
\begin{proof}
    By definition of the reduced problem, the potential term is fixed at $\mathcal V[\phi_n] = \mathcal V_0 < 0$, and the kinetic term is constrained by $\mathcal K[\phi_n] < -\mathcal V_0$. Therefore, it is uniformly bounded from above.

    The uniform boundedness of the gradient term $\mathcal T[\phi_n]$ follows from the assumption that $\mathcal R[\phi_n]$ is bounded from above. From the definition of $\mathcal R$ in Eq.~\eqref{eq:def_R}, we have
    \begin{equation}
        \mathcal R[\phi_n] = \frac{(\mathcal T[\phi_n])^{\frac{D-1}{D-3}}}{-\mathcal V_0 - \mathcal K[\phi_n]} \ge \frac{(\mathcal T[\phi_n])^{\frac{D-1}{D-3}}}{-\mathcal V_0}.
    \end{equation}
    As $\mathcal R[\phi_n]$ is bounded from above by $\mathcal R_0$, it follows that $\mathcal T[\phi_n]$ is also uniformly bounded: $\sup_n \mathcal T[\phi_n] \le (-\mathcal R_0 \mathcal V_0)^{\frac{D-3}{D-1}}$.
\end{proof}
\begin{lemma}\label{lem:bound}
    Let $\phi \in W^{1,2}_{\rm loc}(\mathbb R^{D-1}\times S^1)$ be a function that vanishes at spatial infinity. If the gradient term $\mathcal T[\phi]$ is finite, then $\int\odif[order=D]{z}\ab|\frac{\phi(z)}{|x|}|^2$ and $\int_{\Omega_R}\odif[order=D]{z}|\phi(z)|^2$ are both bounded by $\mathcal T[\phi]$.
\end{lemma}
\begin{proof}
    The first bound is a direct consequence of Hardy's inequality\footnote{This extension of Hardy's inequality to functions that vanish at spatial infinity can be justified by the following argument. For any $\epsilon > 0$, consider the truncated function $\phi_\epsilon = \max(\phi - \epsilon, 0)$. Since $\phi$ vanishes at spatial infinity, $\phi_\epsilon$ has compact support in the $x$ directions and belongs to $W^{1,2}(\mathbb{R}^{D-1} \times S^1)$. The standard Hardy inequality can be applied to $\phi_\epsilon$. The desired inequality for $\phi$ is then recovered by taking the limit $\epsilon \to 0$ and applying the monotone convergence theorem.}:
    \begin{align}
        \int\odif[order=D]{z}\ab|\frac{\phi(z)}{|x|}|^2\leq\ab(\frac{2}{D-3})^2\int\odif[order=D]{z}\ab|\nabla_x\phi(z)|^2=2\ab(\frac{2}{D-3})^2\mathcal T[\phi].
    \end{align}
    The second bound for the local $L^2$ norm on $\Omega_R$ follows from the first:
    \begin{align}
        \int_{\Omega_R}\odif[order=D]{z}|\phi(z)|^2&\leq R^2\int_{\Omega_R}\odif[order=D]{z}\ab|\frac{\phi(z)}{|x|}|^2\leq 2R^2\ab(\frac{2}{D-3})^2\mathcal T[\phi].\label{eq:l2}
    \end{align}
\end{proof}

A direct consequence of Lemma \ref{lem:uniform} and Lemma \ref{lem:bound} is that the minimizing sequence $\{\phi_n\}$ is uniformly bounded in the Sobolev space $W^{1,2}(\Omega_R)$ for any finite $R>0$.

\begin{lemma}\label{lem:Rstar}
    Let $\{\phi_n\}$ be a minimization sequence and $V$ be an admissible potential. Then, there exists $R_*>0$ such that $\int_{\Omega_{R_*}^c}\odif[order=D]{z}V(\phi_n(z))$ is negative for all $n$. 
\end{lemma}
\begin{proof}
    From the admissibility condition on the potential, Eq.~\eqref{eq:condition_v}, we have $V(\phi) + b|\phi|^B \ge a|\phi|^A \ge 0$ for all $\phi \ge 0$.
    This bounds the integral of $V(\phi_n)$ over the exterior domain $\Omega_{R_*}^c$ as follows:
    \begin{align}
        \int_{\Omega_{R_*}^c}\odif[order=D]{z}V(\phi_n(z))&\leq \int_{\Omega_{R_*}^c}\odif[order=D]{z}V(\phi_n(z)) \nonumber\\
        &\hspace{3ex}+ \int_{\Omega_{R_*}}\odif[order=D]{z}\ab[V(\phi_n(z))+b|\phi_n(z)|^B]\nonumber\\
        &=\mathcal V[\phi_n]+b\int_{\Omega_{R_*}}\odif[order=D]{z}|\phi_n(z)|^B.\label{eq:negative_int_v}
    \end{align}
    By definition of the reduced problem, $\mathcal V[\phi_n] = \mathcal V_0 < 0$. To show that the right-hand side is negative for a sufficiently small $R_*$, we need to demonstrate that the second term can be made arbitrarily small, uniformly in $n$, as $R_* \to 0$.

    Let $p_* = 2D/(D-2)$ be the critical Sobolev exponent. Since $B < p_*$, we can apply H\"older's inequality to the integral over $\Omega_{R_*}$:
    \begin{align}
        \int_{\Omega_{R_*}}\odif[order=D]{z}|\phi_n(z)|^B &\leq \ab(\int_{\Omega_{R_*}}\odif[order=D]{z})^{1-\frac{B}{p_*}} \ab(\int_{\Omega_{R_*}}\odif[order=D]{z}|\phi_n(z)|^{p_*})^{\frac{B}{p_*}} \nonumber \\
        &= \ab(\beta \omega_{D-1} R_*^{D-1})^{1-\frac{B}{p_*}} ||\phi_n||_{L^{p_*}(\Omega_{R_*})}^B. \label{eq:holder}
    \end{align}
    The Sobolev embedding theorem for $W^{1,2}(\Omega_{R_*})$ states that there exists a constant $C_S$ such that $||\phi_n||_{L^{p_*}(\Omega_{R_*})} \leq C_S ||\phi_n||_{W^{1,2}(\Omega_{R_*})}$.
    From Lemmas \ref{lem:uniform} and \ref{lem:bound}, the sequence $\{\phi_n\}$ is uniformly bounded in $W^{1,2}(\Omega_{R_*})$, {\it i.e.}, $\sup_n ||\phi_n||_{W^{1,2}(\Omega_{R_*})} \le M$ for some constant $M$.
    Therefore, the $L^{p_*}$ norm is also uniformly bounded: $\sup_n ||\phi_n||_{L^{p_*}(\Omega_{R_*})} \le C_S M$.

    Substituting this back into Eq.~\eqref{eq:holder}, we find
    \begin{equation}
        \int_{\Omega_{R_*}}\odif[order=D]{z}|\phi_n(z)|^B \leq \ab(\beta \omega_{D-1} R_*^{D-1})^{1-\frac{B}{p_*}} (C_S M)^B.
    \end{equation}
    Since $1-B/p_* > 0$, the right-hand side vanishes as $R_* \to 0$, uniformly for all $n$. We can therefore choose $R_* > 0$ small enough such that $b \int_{\Omega_{R_*}}\odif[order=D]{z}|\phi_n(z)|^B < -\mathcal V_0$ for all $n$.
    With this choice of $R_*$, Eq.~\eqref{eq:negative_int_v} becomes negative for all $n$, which completes the proof.
\end{proof}
\begin{lemma}\label{lem:conv}
    There exists a minimizing sequence of spherically symmetric monotone functions $\{\phi_n\}$ and a spherically symmetric monotone function $\phi$ such that $\phi_n$ converges to $\phi$ in $L^q_{\rm loc}(\mathbb R^{D-1}\times S^1)$ for all $0< q<2D/(D-2)$ simultaneously, and $\partial_\tau\phi_n$ and $\nabla_x\phi_n$ weakly converge to $\partial_\tau\phi$ and $\nabla_x\phi$ in $L^2(\mathbb{R}^{D-1} \times S^1)$, respectively.
\end{lemma}
\begin{proof}
    Rellich-Kondrachov's theorem states that $W^{1,2}(\Omega_R)\Subset L^q(\Omega_R)$ for $1\leq q<2D/(D-2)$, {\it i.e.} there exists a subsequence $\{\phi_{n_k}(|x|,\tau)\}$ and a function $\phi(x,\tau)\in L^q(\Omega_R)$ such that $\phi_{n_k}(|x|,\tau)\to\phi(x,\tau)$ as $k\to\infty$ in $L^q(\Omega_R)$. 
    Furthermore, since $\Omega_R$ has finite measure, an application of H\"older's inequality shows that this convergence in $L^q(\Omega_R)$ implies convergence in $L^h(\Omega_R)$ for all $0<h<q$:
    \begin{align}
        ||\phi-\phi_{n_k}||_{L^h(\Omega_R)}^h&=||(\phi-\phi_{n_k})^h||_{L^1(\Omega_R)}\nonumber\\
        &\leq||(\phi-\phi_{n_k})^h||_{L^{\frac{q}{h}}(\Omega_R)}||1||_{L^{\frac{q}{q-h}}(\Omega_R)}\nonumber\\
        &=||\phi-\phi_{n_k}||_{L^q(\Omega_R)}^h||1||_{L^{\frac{q}{q-h}}(\Omega_R)}.
    \end{align}

    By a standard diagonal argument over increasing domains $\Omega_R$ and exponents $q$, we can extract a subsequence that converges to a limit function $\phi$ in $L^q_{\rm loc}(\mathbb{R}^{D-1} \times S^1)$ for all $0 < q < 2D/(D-2)$ simultaneously. The limit function $\phi$ must also be $\mathcal H^D$-a.e.~spherically symmetric and monotonic. This is because the sequence $\{\phi_n\}$ consists of such functions, and these properties are preserved under the almost everywhere convergence implied by the $L^q_{\rm loc}$ limit.

   Regarding the derivatives, Lemma \ref{lem:uniform} establishes that the sequences $\{\nabla_x \phi_{n_k}\}$ and $\{\partial_\tau \phi_{n_k}\}$ are uniformly bounded in $L^2(\mathbb{R}^{D-1} \times S^1)$. Since $L^2$ is a reflexive Banach space, the Banach-Alaoglu theorem and the Eberlein-\v Smulian theorem guarantee the existence of weakly convergent subsequences. Therefore, we can find a subsequence such that $\nabla_x \phi_{n_{k'}}$ and $\partial_\tau \phi_{n_{k'}}$ weakly converge to $\nabla_x \phi$ and $\partial_\tau \phi$ in $L^2(\mathbb{R}^{D-1} \times S^1)$.
\end{proof}

\begin{lemma}\label{lem:limit_finite}
    Let $\phi$ be the limit function of Lemma \ref{lem:conv}. Then, $\mathcal T[\phi]$ and $\mathcal K[\phi]$ are both finite.
\end{lemma}
\begin{proof}
    Let $\{\phi_n\}$ be the minimizing sequence from Lemma \ref{lem:conv}, whose derivatives converge weakly in $L^2(\mathbb{R}^{D-1} \times S^1)$. A standard result in functional analysis is that the squared norm in a Hilbert space is weakly lower semicontinuous. Applying this property to the $L^2$ norms of the derivatives yields:
    \begin{equation}
        \int\odif[order=D]{z}\ab|\nabla_x\phi(z)|^2\leq\liminf_{n\to\infty}\int\odif[order=D]{z}\ab|\nabla_x\phi_n(z)|^2,\label{eq:liminf_t}
    \end{equation}
    and
    \begin{equation}
        \int\odif[order=D]{z}\ab|\partial_\tau\phi(z)|^2\leq\liminf_{n\to\infty}\int\odif[order=D]{z}\ab|\partial_\tau\phi_n(z)|^2.\label{eq:liminf_k}
    \end{equation}
    From Lemma \ref{lem:uniform}, this implies that $\mathcal T[\phi]$ and $\mathcal K[\phi]$ are finite.
\end{proof}

\begin{lemma}
    Let $\{\phi_n\}$ and $\phi$ be the minimizing sequence and the limit function of Lemma \ref{lem:conv}. Then, $\phi$ is the solution of the reduced problem, {\it i.e.} the following conditions hold:
    \begin{align}
        \esssup_{(x,\tau)\in\Omega^c_R}|x|^{\frac{D-3}{2}}\phi(|x|,\tau)&<\infty,~({\rm vanishing~at~spatial~infinity})\\
        \mathcal V[\phi]+\mathcal K[\phi]&<0,~({\rm nontriviality})
    \end{align}
    and
    \begin{equation}
        \mathcal R[\phi]\leq\liminf_{n\to\infty}\mathcal R[\phi_n].~({\rm lower~semicontinuity})\label{eq:lower_semicont}
    \end{equation}
\end{lemma}

Notice that, since $\{\phi_n\}$ is a minimizing sequence of $\mathcal R$, the last condition implies
\begin{equation}
    \mathcal R[\phi]=\lim_{n\to\infty}\mathcal R[\phi_n].
\end{equation}

\begin{proof}
    In the following, we prove each of these conditions in turn.
    \subsubsection*{Step 1: Vanishing at spatial infinity}

    Since $\phi(r,\tau)$ is monotonically decreasing in $r$ $\mathcal{H}^2$-a.e.~$(r,\tau)$, we can bound its essential infimum over the compactified time dimension. For $\mathcal{H}$-a.e.~$r_0 > 0$, we have:
    \begin{align}
        \essinf_{\tau}|\phi(r_0,\tau)|&\leq\frac{1}{\beta}\int_0^\beta\odif{\tau}|\phi(r_0,\tau)|\nonumber\\
        &\leq\frac{1}{\omega_{D-1}r_0^{D-1}\beta}\int_{\Omega_{r_0}}\odif[order=D]{z}|\phi(z)|.\label{eq:essinf_phi}
    \end{align}
    Here, $||\phi||_{L^1(\Omega_{r_0})}$ is bounded\footnote{This argument can be generalized to a uniform bound on $||\phi_n||_{L^1(\Omega_{r_0})}$. For any $r_0 > 0$, an application of H\"older's inequality gives:
    \begin{align}
        \int_{\Omega_{r_0}}\odif[order=D]{z}|\phi_n(z)|&\leq\ab(\int_{\Omega_{r_0}}\odif[order=D]{z}\frac{|\phi_n(z)|^2}{|x|^2})^{1/2}\ab(\int_{\Omega_{r_0}}\odif[order=D]{z}|x|^2)^{1/2},\label{eq:bound_phi}
    \end{align}
    which is uniformly bounded from above by Lemma \ref{lem:bound}.} because of Lemma \ref{lem:conv}.

    For any $r_1 > r_0 > 0$, an application of H\"older's inequality to the annular region $\Omega_{r_0}^c \cap \Omega_{r_1}$ yields:
    \begin{align}
        \int_{\Omega_{r_0}^c\cap\Omega_{r_1}}\odif[order=D]{z}|\partial_\tau\phi(z)|&\leq\ab(\int_{\Omega_{r_0}^c\cap\Omega_{r_1}}\odif[order=D]{z}|\partial_\tau\phi(z)|^2)^{1/2}\ab(\int_{\Omega_{r_0}^c\cap\Omega_{r_1}}\odif[order=D]{z})^{1/2},\label{eq:bound_int_dphi_dtau}
    \end{align}
    which is bounded from above by Lemma \ref{lem:limit_finite}.

    Since both $\essinf_\tau\phi(r,\tau)$ and $\esssup_\tau\phi(r,\tau)$ are $\mathcal H$-a.e.~monotonically decreasing with respect to $r$, we have
    \begin{align}
        &\esssup_\tau\phi(r_1,\tau)\nonumber\\
        &\leq\frac{D-1}{r_1^{D-1}-r_0^{D-1}}\int_{r_0}^{r_1}\odif{r}r^{D-2}\esssup_\tau\phi(r,\tau)\nonumber\\
        &\leq\frac{D-1}{r_1^{D-1}-r_0^{D-1}}\int_{r_0}^{r_1}\odif{r}r^{D-2}\ab[\essinf_\tau\phi(r,\tau)+\int_0^\beta\odif{\tau}|\partial_\tau\phi(r,\tau)|]\nonumber\\
        &\leq\essinf_\tau\phi(r_0,\tau)+\frac{1}{\omega_{D-1}(r_1^{D-1}-r_0^{D-1})}\int_{\Omega_{r_0}^c\cap\Omega_{r_1}}\odif[order=D]{z}|\partial_\tau\phi(z)|,
    \end{align}
    for $\mathcal H$-a.e.~$r_1>r_0$. Now, Eqs.~\eqref{eq:bound_phi} and \eqref{eq:essinf_phi} imply that the first term is controlled by $r_0^{-(D-3)/2}$, while Eq.~\eqref{eq:bound_int_dphi_dtau} implies that the second term is controlled by $(r_1^{D-1}-r_0^{D-1})^{-1/2}$. Thus, setting $r_1=2r_0=R>0$, we obtain
    \begin{equation}
        \esssup_{(x,\tau)\in\Omega^c_{R}}|x|^{\frac{D-3}{2}}\phi(|x|,\tau)<\infty.\label{eq:esssup}
    \end{equation}

    Note that the same argument, applied to each $\phi_n$ in the sequence, establishes the uniform boundedness of $\esssup_{(x,\tau)\in\Omega^c_{R}}|x|^{\frac{D-3}{2}}\phi_n(|x|,\tau)$, a consequence of the uniform bounds on $\mathcal{T}[\phi_n]$ and $\mathcal{K}[\phi_n]$.

    \subsubsection*{Step 2: Non-triviality}
    We begin by decomposing the integral of $|\phi_n|^B$ into contributions from the interior domain $\Omega_R$ and the exterior domain $\Omega_R^c$:
    \begin{equation}
        \int\odif[order=D]{z}|\phi_n(z)|^B=\int_{\Omega_{R}}\odif[order=D]{z}|\phi_n(z)|^B+\int_{\Omega_{R}^c}\odif[order=D]{z}|\phi_n(z)|^B,\label{eq:int_phi_B}
    \end{equation}
    with $R>0$. For the interior domain, since $B$ is less than the critical Sobolev exponent $2D/(D-2)$, Lemma \ref{lem:conv} ensures that $\phi_n$ converges to $\phi$ in $L^B(\Omega_R)$, which implies
    \begin{equation}
        \int_{\Omega_R}\odif[order=D]{z}|\phi(z)|^B=\lim_{n\to\infty}\int_{\Omega_R}\odif[order=D]{z}|\phi_n(z)|^B.\label{eq:conv_phi_B}
    \end{equation}

    For the exterior domain, let $R_*$ be the radius defined in Lemma \ref{lem:Rstar}. Since the integral of the potential in this domain is negative, the admissibility condition on the potential, Eq.~\eqref{eq:condition_v}, implies
    \begin{equation}
        b\int_{\Omega_{R_*}^c}\odif[order=D]{z}|\phi_n(z)|^B\geq a\int_{\Omega_{R_*}^c}\odif[order=D]{z}|\phi_n(z)|^A.\label{eq:b-a}
    \end{equation}
    To control the integral of $|\phi_n|^A$, we first establish a uniform bound on a higher-order moment. Let $C = 2(D-1)/(D-3)>2D/(D-2)>B$. We have
    \begin{align}
        \int_{\Omega_{R_*}^c}\odif[order=D]{z}|\phi_n(z)|^{C}&=\int_{\Omega_{R_*}^c}\odif[order=D]{z}\left||x|^{\frac{D-3}{2}}\phi_n(z)\right|^{\frac{4}{D-3}}\frac{|\phi_n(z)|^2}{|x|^2}\nonumber\\
        &\leq\left|\esssup_{(x,\tau)\in\Omega_{R_*}^c}|x|^{\frac{D-3}{2}}\phi_n(z)\right|^{\frac{4}{D-3}}\int\odif[order=D]{z}\frac{|\phi_n(z)|^2}{|x|^2}.\label{eq:bound_phi_C}
    \end{align}
    As established in Step 1, the essential supremum term is uniformly bounded. The integral term is also uniformly bounded due to Lemma \ref{lem:bound}. Consequently, the integral on the left-hand side is uniformly bounded.
    
    Applying H\"older's inequality, we can relate the integrals of $|\phi_n|^B$ and $|\phi_n|^A$:
    \begin{align}
        \int_{\Omega_{R_*}^c}\odif[order=D]{z}|\phi_n(z)|^B&=\int_{\Omega_{R_*}^c}\odif[order=D]{z}|\phi_n(z)|^{A\frac{C-B}{C-A}}|\phi_n(z)|^{C\frac{A-B}{A-C}}\nonumber\\
        &\leq\ab(\int_{\Omega_{R_*}^c}\odif[order=D]{z}|\phi_n(z)|^A)^{\frac{C-B}{C-A}}\ab(\int_{\Omega_{R_*}^c}\odif[order=D]{z}|\phi_n(z)|^C)^{\frac{B-A}{C-A}}.\label{eq:holder_C}
    \end{align}
    Combining this with Eq.~\eqref{eq:b-a} yields a bound on the integral of $|\phi_n|^A$:
    \begin{equation}
        \int_{\Omega_{R_*}^c}\odif[order=D]{z}|\phi_n(z)|^A\leq\ab(\frac{b}{a})^{\frac{C-A}{B-A}}\int_{\Omega_{R_*}^c}\odif[order=D]{z}|\phi_n(z)|^C.
    \end{equation}

    This allows us to control the tail of the integral of $|\phi_n|^B$. For $R>R_*$, we have
    \begin{align}
        &\int_{\Omega_R^c}\odif[order=D]{z}|\phi_n(z)|^B\nonumber\\
        &=|\esssup_\tau\phi_n(R,\tau)|^{B-A}\int_{\Omega_R^c}\odif[order=D]{z}|\phi_n(z)|^A\nonumber\\
        &\leq|\esssup_\tau\phi_n(R,\tau)|^{B-A}\int_{\Omega_{R_*}^c}\odif[order=D]{z}|\phi_n(z)|^A\nonumber\\
        &\leq R^{-\frac{D-3}{2}(B-A)}\left|\esssup_{(x,\tau)\in\Omega_{R}^c}|x|^{\frac{D-3}{2}}\phi_n(z)\right|^{B-A}\ab(\frac{b}{a})^{\frac{C-A}{B-A}}\int_{\Omega_{R_*}^c}\odif[order=D]{z}|\phi_n(z)|^C\nonumber\\
        &\leq\frac{\mathcal C}{R^{\frac{D-3}{2}(B-A)}}.\label{eq:bound_phi_B}
    \end{align}
    Here, $\mathcal C$ is a constant independent of $n$, which follows from the uniform bounds established in Step 1 and Eq.~\eqref{eq:bound_phi_C}.

    The convergence of $\phi_n$ to $\phi$ in $L^q_{\rm loc}$ for $q>0$ implies convergence in measure on compact sets. A standard result in measure theory then guarantees the existence of a subsequence, which we do not relabel, that converges to $\phi$ for $\mathcal H^D$-almost every $z$. The lower semicontinuity of $V$ allows us to relate the potential of the limit function to the limit of the potentials:
    \begin{equation}
        V(\phi(z))+b|\phi(z)|^B\leq\liminf_{n\to\infty}\ab[V(\phi_n(z))+b|\phi_n(z)|^B],
    \end{equation}
    for $\mathcal H^D$-a.e.~$z$. The admissibility condition on the potential ensures that the integrand $V(\phi_n(z))+b|\phi_n(z)|^B$ is non-negative. Integrating the pointwise inequality and applying Fatou's lemma then yields:
    \begin{equation}
        \mathcal V[\phi]+b\int\odif[order=D]{z}|\phi(z)|^B\leq \liminf_{n\to\infty}\ab[\mathcal V[\phi_n]+b\int\odif[order=D]{z}|\phi_n(z)|^B].\label{eq:fatou}
    \end{equation}
    
    To establish the lower semicontinuity of $\mathcal{V}$, we bound the right-hand side of Eq.~\eqref{eq:fatou}. Using the decomposition from Eq.~\eqref{eq:int_phi_B}, we have:
    \begin{align}
        &\liminf_{n\to\infty}\ab[\mathcal V[\phi_n]+b\int\odif[order=D]{z}|\phi_n(z)|^B]\nonumber\\
        &\leq\liminf_{n\to\infty}\mathcal V[\phi_n]+\limsup_{n\to\infty}b\int\odif[order=D]{z}|\phi_n(z)|^B\nonumber\\
        &\leq\liminf_{n\to\infty}\mathcal V[\phi_n]+\limsup_{n\to\infty}b\int_{\Omega_R}\odif[order=D]{z}|\phi_n(z)|^B+\limsup_{n\to\infty}b\int_{\Omega_R^c}\odif[order=D]{z}|\phi_n(z)|^B\nonumber\\
        &\leq\liminf_{n\to\infty}\mathcal V[\phi_n]+b\int_{\Omega_R}\odif[order=D]{z}|\phi(z)|^B+\frac{b\mathcal C}{R^{\frac{D-3}{2}(B-A)}}.
    \end{align}
    Here, we have used the strong convergence in the interior domain (Eq.~\eqref{eq:conv_phi_B}) and the uniform bound on the exterior tail (Eq.~\eqref{eq:bound_phi_B}). Combining this with Eq.~\eqref{eq:fatou} and taking the limit $R\to\infty$, we obtain the lower semicontinuity of the potential term:
    \begin{equation}
        \mathcal V[\phi]\leq \liminf_{n\to\infty}\mathcal V[\phi_n]=\mathcal V[\phi_0]=\limsup_{n\to\infty}\mathcal V[\phi_n].\label{eq:v_conv}
    \end{equation}

    Combining the lower semicontinuity of the kinetic term, Eq.~\eqref{eq:liminf_k}, with the result for the potential term, Eq.~\eqref{eq:v_conv}, we can establish the non-triviality of the limit function $\phi$:
    \begin{equation}
        \mathcal V[\phi]+\mathcal K[\phi]\leq\mathcal V_0+\liminf_{n\to\infty}\mathcal K[\phi_n]\leq\mathcal V_0+\sup_n\mathcal K[\phi_n]<0.
    \end{equation}

    \subsubsection*{Step 3: Lower semicontinuity}
    Finally, we combine the lower semicontinuity of the gradient and kinetic terms (Eqs.~\eqref{eq:liminf_t} and \eqref{eq:liminf_k}) with the convergence of the potential term (Eq.~\eqref{eq:v_conv}) to establish the lower semicontinuity of the functional $\mathcal R$:
    \begin{equation}
        \mathcal R[\phi]=\frac{(\mathcal T[\phi])^{\frac{D-1}{D-3}}}{-\mathcal V[\phi]-\mathcal K[\phi]}\leq\frac{(\liminf_{n\to\infty}\mathcal T[\phi_n])^{\frac{D-1}{D-3}}}{\limsup_{n\to\infty}(-\mathcal V[\phi_n]-\mathcal K[\phi_n])}\leq\liminf_{n\to\infty}\mathcal R[\phi_n].
    \end{equation}
\end{proof}

\section{Conclusion}

\label{sec:Conclusion}

In this work, we have extended the classical analysis of CGM to finite temperature, providing a rigorous proof that, for a broad class of scalar potentials, the bounce solution is necessarily $O(D-1)$-symmetric and monotonic in the spatial directions. This finding is non-trivial since compactification of the Euclidean time dimension explicitly breaks the full $O(D)$ symmetry, potentially admitting more complex solutions. A rigorous proof is thus essential to establish a firm mathematical foundation for the symmetry properties widely assumed in studies of thermal vacuum decay.

Our methodology hinges on reformulating the search for saddle points of the Euclidean action into an equivalent, genuine minimization problem for a scale-invariant functional, $\mathcal R$. We first proved that a minimizer of this reduced problem corresponds to a non-trivial saddle point of the action. 
Then, by leveraging Steiner symmetrization and extending the recent theorems, we have shown that any minimizer of the reduced problems is $O(D-1)$ symmetric and monotonic in the spatial directions.

Notably, our admissibility conditions on the potential are identical to those imposed by CGM for the zero-temperature case, and our proof naturally recovers the $O(D-1)$ part of the CGM result in the low-temperature limit ($\beta\to\infty$).
This is a direct consequence of the reduced problem approach, which must guard against collapse into highly localized field configurations. In such limits, the size of $S^1$ becomes irrelevant, and the problem effectively reverts to the $D$-dimensional zero-temperature scenario, necessitating the same growth conditions on the potential.  While these conditions are sufficient, they may not be necessary for all temperatures: The existence of bounce solutions under weaker conditions in the high-temperature limit suggests that one may relax the conditions depending on the temperature, though it would likely require a departure from the present reduced problem approach.

From a physical perspective, our findings have direct implications for early universe cosmology. First-order phase transitions, driven by thermal vacuum decay, are central to scenarios such as electroweak baryogenesis and the production of stochastic gravitational wave backgrounds. The symmetry of the bounce solution dictates the decay rate and the most probable shape of nucleated bubbles, which in turn govern the dynamics of bubble growth, collision, and the resulting cosmological signatures. Our work provides a rigorous underpinning for the predictions of these models and, for the first time, justifies the assumption of $O(3)$ symmetry in four dimensions used in the literature.

There are several avenues for future research. The first is the inclusion of gravitational effects. Proving that an $O(D-1)$-symmetric bounce solution minimizes the action in a gravitational theory at finite temperature remains an important open problem, analogous to the long-standing challenge at zero temperature that was recently addressed for the AdS case \cite{Oshita:2023pwr}. Another natural extension is to generalize our results to multi-component scalar fields as in the zero temperature case \cite{Blum:2016ipp}, which would require a careful adaptation of rearrangement inequalities. Finally, relaxing the admissibility conditions on the potential could further broaden the applicability of these results.

\section*{Acknowledgements}
We would like to thank Takato Mori and Naritaka Oshita for useful discussions and comments.
Y.S. is supported by the Slovenian Research Agency under the research grant J1-4389. M.Y. is supported by IBS under the project code, IBS-R018-D3, and by JSPS Grant-in-Aid for Scientific Research Number JP23K20843.

\section*{Data availability}
This study is theoretical and does not generate or analyze any datasets. All results follow from analytical arguments presented in the manuscript.

\section*{Conflict of interest}
The authors declare no competing interests.

\bibliographystyle{unsrt}
\bibliography{steiner}

\end{document}